\documentclass[9pt,conference]{IEEEtran}

\usepackage[utf8]{inputenc}
\usepackage[colorlinks,bookmarksopen,bookmarksnumbered,citecolor=black,urlcolor=black,linkcolor=black]{hyperref}
\usepackage{mathtools}
\usepackage{amsthm}
\usepackage{amssymb}
\usepackage{amsmath}
\usepackage{yhmath}
\usepackage{amsfonts}
\usepackage{url}
\usepackage{listings}
\usepackage{color}
\usepackage{xcolor}
\usepackage{graphicx}
\usepackage{caption}
\usepackage[noadjust]{cite}
\usepackage{breqn}
\usepackage{makecell}
\usepackage{enumitem}
\usepackage[titlenumbered,ruled,linesnumbered]{algorithm2e}
\usepackage{balance}
\usepackage{amssymb}
\usepackage{multirow}
\usepackage{multicol} 
\usepackage{soul}
\usepackage{subfigure}
\usepackage{tikz}
\usepackage{makecell,rotating,multirow,diagbox}
\usepackage{booktabs}  
\usepackage{multirow}  
\usepackage{diagbox}   
\usepackage{siunitx}   
\usepackage{bm}

\newtheorem{theorem}{Theorem}
\newtheorem{lemma}{Lemma}

\newcommand{\blackcircle}[1]{%
\tikz[baseline=(char.base)]{
\node[shape=circle,draw,inner sep=0pt,fill=black, text=white, minimum size=1.1em] (char) {#1};}}

\newcommand{\pdag}{{$p$-DAG}\xspace}

\newcommand{\eg}{{\it e.g.}\xspace}
\newcommand{\ie}{{\it i.e.}\xspace}

\newcommand{\metric}{NOAR\xspace}

\begin{document}

\title{Exploiting the Uncertainty of the Longest Paths:\\Response Time Analysis for Probabilistic DAG Tasks}



\author{
\IEEEauthorblockN{Yiyang Gao\IEEEauthorrefmark{2},Shuai Zhao\IEEEauthorrefmark{2}, Boyang Li\IEEEauthorrefmark{2}, Xinwei Fang\IEEEauthorrefmark{3}, Zhiyang Lin\IEEEauthorrefmark{2}, Zhe Jiang\IEEEauthorrefmark{4},Nan Guan\IEEEauthorrefmark{6}}
\IEEEauthorblockA{
\IEEEauthorrefmark{2}Sun Yat-sen University, China
\IEEEauthorrefmark{3}University of York, UK
\IEEEauthorrefmark{4}Southeast University, China
\IEEEauthorrefmark{6}City University of Hong Kong, Hong Kong SAR
}
}

\maketitle

\begin{abstract}
Parallel real-time systems (\textit{\textbf{e.g.,}} autonomous driving systems) often contain functionalities with complex dependencies and execution uncertainties, leading to significant timing variability which can be represented as a probabilistic distribution. 
However, existing timing analysis either produces a single conservative bound or suffers from severe scalability issues due to the exhaustive enumeration of every execution scenario. This causes significant difficulties in leveraging the probabilistic timing behaviours, resulting in sub-optimal design solutions. 
Modelling the system as a probabilistic directed acyclic graph ($\bm{p}$-DAG), 
this paper presents a probabilistic response time analysis based on the longest paths of the $\bm{p}$-DAG across all execution scenarios, enhancing the capability of the analysis by eliminating the need for enumeration. 
We first identify every longest path based on the structure of $\bm{p}$-DAG and compute the probability of its occurrence.
Then, the worst-case interfering workload is computed for each longest path, forming a complete probabilistic response time distribution with correctness guarantees.
Experiments show that compared to the enumeration-based approach, the proposed analysis effectively scales to large $\bm{p}$-DAGs with computation cost reduced by six orders of magnitude while maintaining a low deviation (1.04\% on average and below 5\% for most $\bm{p}$-DAGs), empowering system design solutions with improved resource efficiency.
\end{abstract}

\IEEEpeerreviewmaketitle

\section{Introduction}
\label{sec:intro}



The parallel tasks in multicore real-time systems (\eg, automotive, avionics and robotics) often contain complex dependencies~\cite{zhao2020dag,zhao2023cache,he2019intra,he2021response}, and more importantly, execution uncertainties~\cite{baruah2021feasibility,ueter2021response,melani2015response}, \eg, the ``if-else" statements that execute different branches under varied conditions.
Such execution uncertainties widely exist both within the execution of a single task and between parallel tasks.
This leads to significant variability in the timing behaviours of the system, which can be represented as a probabilistic timing distribution~\cite{nelis2016variability,cazorla2013proartis}. 

Most design and verification methods consider task dependencies by modelling the system as a directed acyclic graph (DAG)~\cite{graham1969bounds,zhao2020dag,he2019intra,he2021response,zhao2022dag}. 
However, these methods often apply a single conservative timing bound (\ie, the worst-case response time) regardless of the execution uncertainties within the DAG~\cite{zhao2020dag,he2019intra}.  
As systems become ever-complex and non-deterministic, such methods are less effective due to insufficient and overly pessimistic analytical results.
For instance, the automotive standard ISO-26262 defines a failure rate for each automotive safety integrity level~\cite{yi2024can,iso201126262,birch2013safety,palin2011iso}, demanding a probabilistic timing analysis that empowers more informative decision-making and design solutions for automotive systems.

Numerous studies have been conducted on the probabilistic worst-case execution time (WCET) of a single thread~\cite{davis2019survey,abella2014comparison,bernat2002wcet,bozhko2023really}. 
However, limited results are reported that address DAG tasks with probabilistic executions (\pdag{})~\cite{ueter2021response}, which are commonly found in real-world applications such as autonomous driving systems~\cite{houssam2020hpc,zhao2024risk,maier2023causal}.
Fig.~\ref{figs:CDAG_example} presents an example \pdag{} with two probabilistic structures, each containing branches with different execution probabilities. In each release, only one branch of every probabilistic structure can execute, yielding different DAG structures with varied response times.

For a \pdag{} modelled system, the existing method~\cite{ueter2021response} produces its probabilistic response time distribution by enumerating through all execution scenarios, with Graham's bound~\cite{graham1969bounds} applied to analyse the traditional DAG of each scenario.
Fig.~\ref{fig:1b} illustrates the cumulative probability distribution of the response times for the example \pdag{}.
However, this approach suffers from severe scalability issues due to the need for exhaustive recursions of every execution scenario, which fails to provide any results for large and complex \pdag{}s. This significantly limits its applicability, imposing huge challenges for the design and verification of systems with \pdag{}s.

\begin{figure}[t]
\centering
\subfigbottomskip=1pt 
\subfigure[The structure of a \pdag.]{\label{figs:CDAG_example} 
\includegraphics[width=.54\columnwidth]{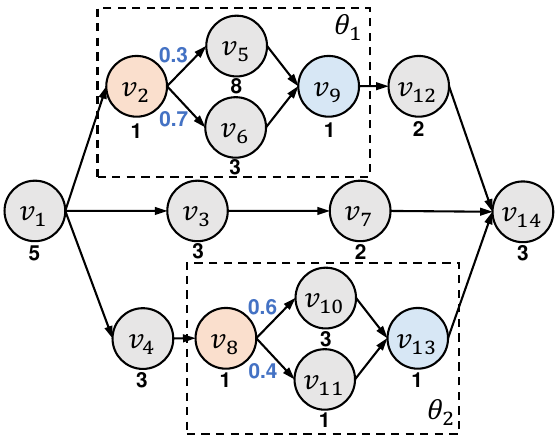}}
\subfigure[Response time distribution.]{\label{fig:1b} 
\includegraphics[width=.42\columnwidth]{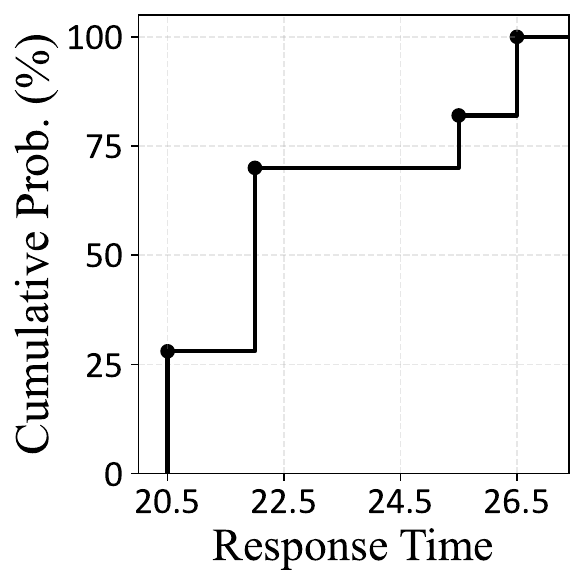}}
\caption{A \pdag with two probabilistic structures and its cumulative probability distribution in response time \textit{(numbers in black: worst-case execution time; numbers in blue: execution probability)}.}
\vspace{-5pt}
\label{figs:exp_dis}
\end{figure}

\textbf{Contributions.} 
This paper presents a probabilistic response time analysis of a \pdag{}, which eliminates the need for enumeration by leveraging the longest paths across all the execution scenarios.
To achieve this, (i) we first identify the set of longest paths of the \pdag{} by determining a lower bound on the longest paths of the \pdag{} and its sub-structures. 
(ii) For each identified path, an analysis is constructed to compute the probability where the path is executed and is the longest path.
(iii) Finally, the worst-case interfering workload associated with each longest path is computed, forming a complete probabilistic response time distribution of the \pdag{} with correctness guarantees.
(iv) Experiments show that the proposed analysis effectively addresses the scalability issue, enhancing its capability in analysing large and complex \pdag{}s.
Compared to the enumeration-based approach, the proposed analysis reduces the computation cost by six orders of magnitude while maintaining a deviation of only 1.04\% on average (below 5\% for most \pdag{}s). 
More importantly, we demonstrate that given a specific time limit, the proposed analysis effectively empowers design solutions with improved resource efficiency, including systems with large \pdag{}s.

To the best of our knowledge, this is the first probabilistic response time analysis for \pdag{}s. 
Notably, the analysis can be used in combination with a probabilistic WCET analysis (\eg, the one in~\cite{ziccardi2015epc}), in which a node with a probabilistic WCET can be effectively modelled as a probabilistic structure with the associated WCETs and probabilities.
With the \pdag{} analysis, the probabilistic timing behaviours of parallel tasks can be fully leveraged to produce flexible and optimised design solutions (\eg, through feedback-based online decision-making) that can meet specific timing requirements. 
\section{Task Model and Preliminaries}

This work focuses on the analysis of a periodic \pdag{} running on $m$ symmetric cores.
Below we provide the task model of a \pdag{}, the existing response time analysis, and the motivation of this work.

\subsection{Task Model of a \pdag{}}
\label{sec:task}
As with the traditional DAG model~\cite{zhao2020dag,he2019intra}, a \pdag{} task is defined by $\tau=\{ \mathcal{V}, \mathcal{E}, T, D \}$, where $\mathcal{V}$ represents a set of nodes, $\mathcal{E}$ denotes a set of directed edges, $T$ is the period and $D$ is the deadline. 
A node $v_i \in \mathcal{V}$ indicates a series of computations that must be executed sequentially. The worst-case execution time (WCET) of $v_i$ is defined as $C_i$\footnote{Nodes with probabilistic WCETs can be supported in this work, in which each of such nodes is modelled as a probabilistic structure in the \pdag{}.}.
An edge $e_{i,j} \in \mathcal{E}$ indicates the execution dependency between $v_i$ and $v_j$, \ie, $v_j$ can start only after the completion of $v_i$. 
As with~\cite{zhao2020dag,he2023real}, we assume that $\tau$ has one source node $v_{src}$ and one sink node $v_{snk}$, \ie, $e_{i,src}$ and $e_{snk,i}$ do not exist.
A path $\lambda_{h}=\{ v_{src},..., v_{snk} \}$ is a sequence of nodes in which every two consecutive nodes are connected by an edge.
The set of all paths in the \pdag{} $\tau$ is denoted as $\Lambda$.

In addition, a \pdag{} contains $n$ probabilistic structures $\Theta = \{\theta_1, \theta_2,..., \theta_n\}$.
For a probabilistic structure $\theta_x \in \Theta$, it has an entry node $v^{src}_x$, an exit node $v^{snk}_x$, and a set of probabilistic branches $\{\theta_x^1, ... , \theta_x^k\}$ in between. 
A branch $\theta_x^k \in \theta_x$ is a sub-graph consisting of non-conditional nodes only, \ie, they are either executed unconditionally or not being executed at all. 
For $\theta_1$ of the \pdag{} in Fig.~\ref{figs:CDAG_example}, it starts from $v_2$ and ends at $v_9$, with two branches $\theta^1_1=\{v_5\}$ and $\theta^2_1=\{v_6\}$.
Function $F(\theta^k_x)$ provides the execution probability of the $k\textsuperscript{th}$ branch in $\theta_x$, which can be obtained based on measurements and the analysing methods in~\cite{yi2024can,bozhko2023really,abella2014comparison}. For $\theta_x$, it follows $\sum_{\forall \theta_x^k \in \theta_x} F(\theta^k_x) = 1$, \eg, $F(\theta^1_1)=0.3$ and $F(\theta^2_1)=0.7$ in Fig.~\ref{figs:CDAG_example}.
Notation $|\theta_x|$ gives the number of branches in $\theta_x$.

Depending on the branch $\theta_x^k$ being executed in each $\theta_x$, $\tau$ can release a series of jobs with different non-conditional graphs. 
For instance, the \pdag{} in Fig.~\ref{figs:CDAG_example} can yield jobs with four different graphs.
Notation $\mathcal{G}=\{\mathcal{G}_1, \mathcal{G}_2, ... \}$ denotes the set of unique non-conditional graphs that can be released by $\tau$.
For a $\mathcal{G}_u$, $len(\mathcal{G}_u)$ is the length of its longest path, $vol(\mathcal{G}_u)$ is the workload, 
$H(\mathcal{G}_u)$ provides the set of branches being executed in $\mathcal{G}_u$, and $S(\mathcal{G}_u)$ provides the set of probabilistic structures of $H(\mathcal{G}_u)$.
For all $\mathcal{G}_u \in \mathcal{G}$, $\Lambda^* = \{\lambda_1, \lambda_2,...\}$ denotes the set of unique longest paths in the graphs.

\subsection{Response Time Analysis for \pdag{}s}
Most existing analysis~\cite{zhao2020dag,melani2015response,he2023real,melani2016schedulability,he2019intra} provides a single bound on the response time.
For instance, the Graham's bound~\cite{graham1969bounds} computes the response time $R$ of a given $\tau$ by $R \leq len(\tau) + \frac{1}{m}(\sum_{v_i\in \mathcal{V}} C_i-len(\tau))$, where $len(\tau)$ denotes the longest path among all paths in $\tau$.
However, these methods neglect the execution variability in \pdag{}s, in which only one $\theta_x^k$ can be executed in each $\theta_x$ with a probability of $F(\theta_x^k)$.
For instance, the \pdag{} in Fig.~\ref{fig:1b} has over 70\% and 25\% of probability to finish within time 26 and 21, respectively.
Hence, the traditional analysis fails to depict such variability of the worst case for a \pdag{}, resulting in a number of design limitations such as resource over-provisioning given predefined time limits (see the automotive example in Sec.~\ref{sec:intro}). 

To account for such uncertainties in a \pdag{},
an enumeration-based approach is presented in~\cite{ueter2021response}, which produces the probabilistic response time distribution by iterating through each $\mathcal{G}_u \in \mathcal{G}$.
For a $\mathcal{G}_u$, the response time is computed by Graham's bound with a probability of occurrence computed by $\prod_{\theta_x^k \in H(\mathcal{G}_u)} F(\theta_x^k)$, \ie, the probability of every $\theta_x^k \in H(\mathcal{G}_u)$ being executed.
With the response time and probability of every $\mathcal{G}_u$ determined, the complete probabilistic response time distribution of a \pdag{} can be established, \eg, the one in Fig.~\ref{fig:1b}.
However, by enumerating every $\mathcal{G}_u$ of a \pdag{}, this method incurs significant complexity and computation cost, leading to severe scalability issues that undermine its applicability, especially for large \pdag{}s (see Sec.~\ref{sec:exp} for experimental results).


To address the above issues, this paper proposes a new analysis for a \pdag{} by exploiting the set of the longest paths, providing tight bounds on the response times and their probabilities while eliminating the need for enumeration. 
To achieve this, a method is constructed that identifies the exact set of longest paths (\ie, $\Lambda^*$) across all $\mathcal{G}_u \in \mathcal{G}$ (Sec.~\ref{sec:identify}).
For each $\lambda_h \in \Lambda^*$, Sec.~\ref{sec:prob} computes the probability where $\lambda_h$ is executed and is the longest path. 
Finally, the probabilistic response time distribution of $\tau$ can be constructed with the worst-case interfering workload of each $\lambda_h \in \Lambda^*$ determined (Sec.~\ref{sec:interference}). 

\section{Identifying the Longest Paths of a \pdag{}}
\label{sec:identify}

Existing methods determine the $\Lambda^*$ of $\tau$ by enumerating every $\mathcal{G}_u \in \mathcal{G}$~\cite{ueter2021response}, which significantly intensifies the complexity of analysing \pdag{}s.
In addition, this would result in redundant computations as certain graphs might have the same longest path.
For instance, Fig.~\ref{fig:2a} shows a \pdag{} with the longest path of $\{v_1, v_2,v_5,..., v_{14} \}$ regardless of whether $v_{10}$ or $v_{11}$ is executed.
To address this, this section presents a method that computes $\Lambda^*$ based on the lower bound on all the longest paths of a given \pdag{} $\tau$ and its sub-structures.
To achieve this, a function $\Delta(\tau)$ is constructed that produces the lower bound on paths in $\Lambda^*$ of $\tau$ (Sec.~\ref{sec:MLP}).
Then, we show that $\Delta(\cdot)$ can be effectively applied on $\tau$ and its sub-structures to identify $\Lambda^*$ without the need for enumeration (Sec.~\ref{sec:Threshold}).

\begin{figure}[t]
\vspace{-6pt}
\centering
\subfigbottomskip=1pt 
\subfigure[The path with $v_5$ is the longest path regardless of $v_{10}$ or $v_{11}$ is executed.]{\label{fig:2a} 
\includegraphics[width=.52\columnwidth]{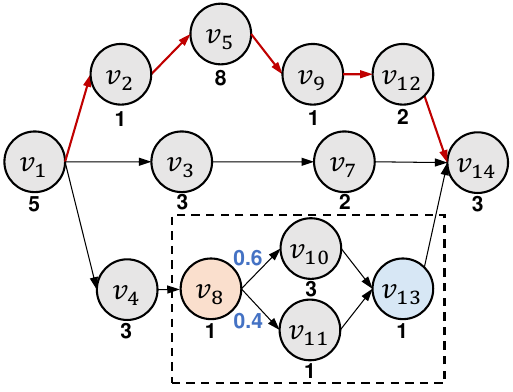}}
\subfigure[{A scenario that every path in $\theta_1$ is longer than paths in $\theta_2$.}]{\label{fig:2b} 
\includegraphics[width=.40\columnwidth]{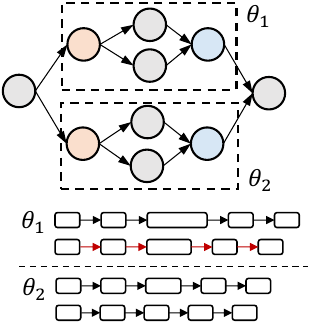}}
\caption{{The illustrative examples used in Sec.~\ref{sec:identify}.}}
\vspace{-8pt}
\label{figs:instances}
\end{figure}

\subsection{Determining the Lower Bound of $\Lambda^*$ for a \pdag{}}\label{sec:MLP}

To compute $\Delta(\tau)$, we first identify the graph (denoted as $\mathcal{G}^\diamond$) that has the minimum longest path in $\mathcal{G}$, \ie, $\Delta(\tau) = len(\mathcal{G}^\diamond) \leq len(\mathcal{G}_u), \forall \mathcal{G}_u \in \mathcal{G}$. 
For a $\tau$, its $\mathcal{G}^\diamond$ can be identified by Theorem~\ref{the:diamond_pragh}. 
\begin{theorem}\label{the:diamond_pragh}
Let $\theta_x^\diamond$ denote the branch executed in $\theta_x$ under $\mathcal{G}^\diamond$, 
it follows that $len(\theta_x^\diamond) \leq len(\theta_x^k), \forall{\theta_x^k \in \theta_x}, \forall{\theta_x \in \Theta}$.
\end{theorem}
\begin{proof}
Suppose there exists a graph $\mathcal{G}_u$ such that $len(\mathcal{G}^\diamond) > len(\mathcal{G}_u)$.
In this case, 
there exists at least one $\theta_x \in \Theta$ in which the branch being executed in $\mathcal{G}_u$ (say $\theta_x^k$) follows $len(\theta_x^\diamond) > len(\theta_x^k)$.
However, this contradicts with the assumption that $\theta^\diamond_x$ has the lowest length among all branches in $\theta_x$.
Hence, the theorem holds.
\end{proof}

We note that Theorem~\ref{the:diamond_pragh} is a sufficient condition for $\mathcal{G}^\diamond$. 
If $\tau$ has a single non-conditional longest path, then $\Delta(\tau)=len(\mathcal{G}^\diamond) = len(\mathcal{G}_u), \forall \mathcal{G}_u \in \mathcal{G}$ regardless of the branch being executed in each $\theta_x$. 
However, this does not undermine computations of $\Delta(\tau)$ based on Theorem~\ref{the:diamond_pragh} and the identification of $\Lambda^*$ based on $\Delta(\tau)$.



With Theorem~\ref{the:diamond_pragh}, Alg.~\ref{algs:neg_path} computes $\Delta(\tau)$ by constructing $\mathcal{G}^\diamond$ based on $\{\mathcal{V}, \mathcal{E}, \Theta\}$ of $\tau$.
The algorithm first initialises $\mathcal{V}^\diamond$ with all the non-conditional nodes in $\mathcal{V}$, which are executed under any graph of $\tau$ (line 1).
Then, for each $\theta_x$ in $\Theta$, the shortest branch $\theta_x^\diamond$ is identified, and the corresponding nodes are added to $\mathcal{V}^\diamond$ (lines 3-6).
Based on $\mathcal{V}^\diamond$, the set of edges that connect these nodes is obtained as $\mathcal{E}^\diamond$ (line 7).
Finally, with $\mathcal{G}^\diamond=\{\mathcal{V}^\diamond, \mathcal{E}^\diamond\}$ constructed, $\Delta(\tau)$ is computed at line 8 using the Deep First Search in linear complexity~\cite{he2019intra}.
For the \pdag{} in Fig.~\ref{figs:CDAG_example}, the $\mathcal{G}^\diamond$ is obtained when $v_6$ and $v_{11}$ are executed, with a longest path of $\{v_1, v_2,v_6,...,v_{14}\}$ and $\Delta(\tau)= 15$.

\begin{algorithm}[t]
\caption{Computation of $\Delta(\tau)$ for a \pdag{} $\tau$} 
\label{algs:neg_path}
$\mathcal{V}^\diamond = \{v_i ~|~ v_i \in \mathcal{V} \wedge v_i \notin \theta_x, \forall \theta_x \in \Theta \}$;\\
{\scriptsize\ttfamily/* Identify nodes in $\mathcal{G}^\diamond$ by Theorem~\ref{the:lambda_neg}. */} \\
\For{$\theta_x \in \Theta$}{
    $\theta^\diamond_x =  \mathop{argmin}_{\theta_x^k}\{ len(\theta_x^k) ~|~ \theta_x^k \in \theta_x \}$;\\
    $\mathcal{V}^\diamond =  \mathcal{V}^\diamond \cup \theta^\diamond_x$;
}

$\mathcal{E}^\diamond = \{e_{i,j} ~|~ v_i,v_j \in \mathcal{V}^\diamond \wedge e_{i,j} \in \mathcal{E} \}$;\\

\Return $len(\mathcal{G}^\diamond = \{\mathcal{V}^\diamond, \mathcal{E}^\diamond\})$;\\
\end{algorithm}

\subsection{Identifying $\Lambda^*$ of a \pdag{} based on $\Delta(\tau)$}\label{sec:Threshold}

With function $\Delta(\tau)$, 
Alg.~\ref{algs:filter} is constructed to compute the $\Lambda^*$ of $\tau$.
Essentially, the algorithm takes $\{\mathcal{V}, \mathcal{E}, \Theta, \Lambda\}$ of $\tau$ as the input, and obtains $\Lambda^*$ by removing the paths that are always dominated by a longer one. 
First, the algorithm identifies the candidates of $\Lambda^*$ by removing paths that are shorter than $\Delta(\tau)$ (line 2) based on Lemma~\ref{the:lambda_neg}.
This effectively speeds up the algorithm, in which only the candidate paths will be examined in later computations. 
\begin{lemma} \label{the:lambda_neg}
For any $\lambda_h \in \Lambda^*$ of $\tau$, it follows that $len(\lambda_h) \geq \Delta(\tau)$.
\end{lemma}
\begin{proof}
This follows directly from Theorem~\ref{the:diamond_pragh} and Alg.~\ref{algs:neg_path}.
\end{proof}

With the candidate paths identified, the algorithm further determines whether a path is always dominated by another path across all scenarios (lines 3-15). 
For two candidate paths $\lambda_a$ and $\lambda_b$, their executions can be categorised into the following three scenarios. 
For S1, there exists no dominance relationship between $\lambda_a$ and $\lambda_b$ as they are not executed simultaneously in any $\mathcal{G}_u$.
Below we focus on determining whether $\lambda_a$ always dominates $\lambda_b$ under S2 and S3.
\begin{itemize}
\item[S1.] $\lambda_a$ and $\lambda_b$ {are never executed} in the same $\mathcal{G}_u$, \ie, a $\theta_x \in S(\lambda_a) \cap S(\lambda_b)$ exists such that $\theta_x^{\alpha} \in H(\lambda_a) \wedge \theta_x^{\beta} \in H(\lambda_b)$;
\item[S2.] $\lambda_a$ and $\lambda_b$ are {always executed or not simultaneously} in any $\mathcal{G}_u \in \mathcal{G}$, \ie, $H(\lambda_a) = H(\lambda_b)$;
\item[S3.] $\lambda_a$ and $\lambda_b$ {can be executed} in the same or in different graphs, \ie, for every $\theta_x \in S(\lambda_a) \cap S(\lambda_b)$, $\theta_x^k \in H(\lambda_a) \cap H(\lambda_b)$.
\end{itemize}


For S2, $\lambda_b$ is not the longest in any $\mathcal{G}_u$ if $len(\lambda_a)> len(\lambda_b)$, as described in Lemma~\ref{lem:identical}.
Following this, the algorithm removes $\lambda_b$ from $\Lambda^*$ if the conditions in Lemma~\ref{lem:identical} hold (lines 5-7).
\begin{lemma}\label{lem:identical}
    For two paths $\lambda_a$ and $ \lambda_b$ with $H(\lambda_a)=H(\lambda_b)$, $\lambda_b \notin \Lambda^*$ if $len(\lambda_a) > len(\lambda_b)$.
\end{lemma}
\begin{proof}
If $H(\lambda_a)=H(\lambda_b)$, $\lambda_a$ and $\lambda_b$ are always executed or not simultaneously under any $\mathcal{G}_u \in \mathcal{G}$. Thus, if $len(\lambda_a) \geq len(\lambda_b)$, $\lambda_b$ is constantly dominated by $\lambda_a$. Hence, the lemma follows. 
\end{proof}

For S3, it is challenging to directly determine the dominance relationship between $\lambda_a$ and $\lambda_b$, as they can be executed in different graphs.
However, if $\lambda_a$ is not executed, an alternative path with the same $S(\lambda_a)$ will be executed, as one branch of each $\theta_x$ must be executed in a graph.
Thus, if $\lambda_b$ is shorter than any of such paths, it is not the longest path in any graph.
For instance, Fig.~\ref{fig:2b} presents a \pdag{} in which the shortest path in $\theta_1$ is longer than any path in $\theta_2$; hence a path that goes through $\theta_2$ is always dominated.

To determine whether $\lambda_b$ is dominated under S3, a sub-structure of $\tau$ is constructed based on $\lambda_a$ and its alternative paths, denoted by $\tau_a=\{\mathcal{V}_a, \mathcal{E}_a, \Theta_a \}$ (lines 10-11). 
First, $\Theta_a$ is constructed by the unique probabilistic structures of $\lambda_a$ \ie, $\theta_x \in S(\lambda_a) \wedge \theta_x \notin S(\lambda_b)$. Then, $\mathcal{V}_a$ and $\mathcal{E}_a$ are computed by nodes in $\lambda_a$ and $\Theta_a$.
Based on $\Delta(\tau_a)$, Lemma~\ref{the:asynch} describes the case where $\lambda_b$ is dominated under S3, and hence, is removed from $\Lambda^*$ by the algorithm (lines 12-14).
\begin{lemma}\label{the:asynch}
For $\lambda_a$ and $\lambda_b$ with $\theta_x^k \in H(\lambda_a) \cap H(\lambda_b), \forall \theta_x \in S(\lambda_a) \cap S(\lambda_b)$, 
it follows that $\lambda_b \notin \Lambda^*$ if $\Delta(\tau_a) > len(\lambda_b)$.
\end{lemma}
\begin{proof}
Suppose that $\lambda_b \in \Lambda^*$ follows $\Delta(\tau_a)>len(\lambda_b)$.
In this case, there always exists a longer path as long as $\lambda_b$ is executed, \ie, a path in $\tau_a$ with a length equal to or higher than $\Delta(\tau_a)$.
Hence, $\lambda_b$ is not the longest in any $\mathcal{G}_u \in \mathcal{G}$, which contradicts with the assumption that $\lambda_b \in \Lambda^*$, and hence, the lemma holds.
\end{proof}

\begin{algorithm}[t]
\caption{Calculation of $\Lambda^*$}\label{algs:filter}

{\scriptsize\ttfamily/* Identify the candidates of $\Lambda^*$ by Lemma~\ref{the:lambda_neg}*/} \\

$\Lambda^* = \{ \lambda_h ~|~ \lambda_h \in \Lambda \wedge len(\lambda_h) \geq \Delta(\mathcal{V},\mathcal{E},\Theta) \}$\;

\For{$\lambda_a, \lambda_b \in \Lambda^*$}{

{\scriptsize\ttfamily/* Determine the dominant path of S2 by Lemma~\ref{lem:identical} */} \\
\If{$H(\lambda_a) = H(\lambda_b) \wedge len(\lambda_a) \geq len(\lambda_b)$}{
{$\Lambda^* = \Lambda^* \setminus  \lambda_b$;}
}

{\scriptsize\ttfamily/* Determine the dominant path of S3 by Lemma~\ref{the:asynch} */} \\
\If{\text{\normalfont{Lemma~\ref{the:asynch} holds for} $\lambda_a$ and $\lambda_b$}}{
$\Theta_a= 
\{\theta_x ~|~ \theta_x \in S(\lambda_a) \wedge \theta_x \notin S(\lambda_b)  \}$;\\
$\mathcal{V}_a = \lambda_a \cup \Theta_a$;~~~~$\mathcal{E}_a=\{e_{i,j} ~|~ v_i, v_j \in \mathcal{V}_a \wedge e_{i,j} \in \mathcal{E} \}$;\\

\If{$\Delta(\tau_a=\{ \mathcal{V}_a, \mathcal{E}_a, \Theta_a \}) \geq len(\lambda_b)$}{
    $\Lambda^* = \Lambda^* \setminus  \lambda_b$;
}
}
}

\Return $\Lambda^*$\;
\end{algorithm}

After every two candidate paths are examined, the algorithm terminates with $\Lambda^*$ of $\tau$ returned (line 17). 
For the \pdag{} in Fig.~\ref{figs:CDAG_example}, $|\Lambda^*|=3$ with $\lambda_1 = \{ v_1,v_2,v_5,v_9,v_{12},v_{14} \}$, $\lambda_2 = \{ v_1,v_2,v_6,v_9,v_{12},v_{14} \}$ and $\lambda_3 = \{ v_1,v_4,v_8,v_{10},v_{13},v_{14} \}$. 
More importantly, by utilising $\Delta(\cdot)$ and the relationship between paths, Alg.~\ref{algs:filter} identifies the exact set of longest path of a \pdag{} without the need for enumerating through every $\mathcal{G}_u \in \mathcal{G}$, as shown in Theorem~\ref{the:complete}. This provides the foundation for the constructed analysis of \pdag{} and the key of addressing the scalability issue.
   

\begin{theorem}\label{the:complete}
Alg.~\ref{algs:filter} produces the exact $\Lambda^*$ of a \pdag{} $\tau$.
\end{theorem}
\begin{proof}
First, for a path $\lambda_b \notin \Lambda^*$, there always exists a longer path $\lambda_a \in \Lambda^*$ under any $\mathcal{G}_u \in \mathcal{G}$. This is guaranteed by Lemmas~\ref{the:lambda_neg},~\ref{lem:identical} and~\ref{the:asynch}, which removes $\lambda_b$ from $\Lambda^*$ if it is lower than $\Delta(\tau)$, $\lambda_a$ or $\Delta(\tau_a)$, respectively.
Second, for any path $\lambda_a \in \Lambda^*$, there exists a graph $\mathcal{G}_u \in \mathcal{G}$ in which $\lambda_a$ is the longest. 
Assume that $\lambda_a \in \Lambda^*$ is not the longest in any graph, $\lambda_a$ is always dominated by a path $\lambda_b \in \Lambda^*$ or its alternative paths in $\tau_b$, and hence, will be removed based on the lemmas. 
Therefore, the theorem holds.
\end{proof}


The time complexity for computing $\Lambda^*$ is $\mathcal{O}(n^4)$. First, Alg.~\ref{algs:neg_path} has a $\mathcal{O}(n^2)$ complexity, 
which 
examines each $\theta^k_x$ of every $\theta_x \in \Theta$.
For Alg.~\ref{algs:filter}, at most $|\Lambda|^2$ iterations are performed to examine the paths, where each iteration can invoke Alg.~\ref{algs:neg_path} once.
Hence, the complexity of Alg.~\ref{algs:filter} is $\mathcal{O}(n^4)$. 
In addition, we note that a number of optimisations can be conducted to speed up the computations, \eg, $\lambda_a$ can be removed directly if $len(\lambda_a)\leq len(\lambda_b)$ at lines 5-6. However, such optimisations are omitted to ease the presentation. 
\section{Probabilistic Analysis of the Longest Paths}
\label{sec:prob}

This section computes the probability where a path $\lambda_h \in \Lambda^*$ is executed and is the longest, denoted as $P(\lambda_h)$.
Such a case can occur if (i) $\lambda_h$ is executed and (ii) all the longer paths in $\Lambda^*$ (say $\lambda_l$) are not executed.
The first part is calculated as $\prod_{\theta_x^k \in H(\lambda_h)}F(\theta_x^k)$, \ie, the probability of every branch $\theta_x^k$ in $H(\lambda_h)$ is executed.
However, it is challenging to obtain the probability of the second part,
in which a path $\lambda_l$ is not executed if any $\theta_x^k \in H(\lambda_l)$ is not being executed.
Hence, the computation of such a probability can become extremely complex when multiple long paths are considered, especially when these paths share certain $\theta_x^k$.
Considering the above, we develop an analysis that produces tight bounds on $P(\lambda_h)$ by leveraging the relationships between $P(\lambda_h)$ of all paths in $\Lambda^*$ (Sec.~\ref{sec:4a}). Then, we demonstrate that the pessimism would not significantly accumulate along with the computation of $P(\lambda_h)$ for every $\lambda_h \in \Lambda^*$, and prove the correctness of the constructed analysis (Sec.~\ref{sec:4c}). 

\subsection{Computation of $P(\lambda_h)$}\label{sec:4a}

The computation of $P(\lambda_h)$ is established based on the following relationship between $P(\lambda_h), \forall \lambda_h \in \Lambda^*$.
First, given that $\Lambda^*$ provides the exact set of the longest paths of $\tau$ (see Theorem~\ref{the:complete}), it follows that $\sum_{\lambda_h \in \Lambda^*} P(\lambda_h) = 1$. In addition, as only one path is the longest path in any $\mathcal{G}_u$, the probability of both $\lambda_a$ and $\lambda_b$ being executed as the longest in one graph is $P(\lambda_a \oplus \lambda_b) = 0$.

Based on the above, $P(\lambda_h)$ can be determined by the sum of probabilities of all other paths in $\Lambda^*$, as shown in Eq.~\ref{eqs:target}. The paths in $\Lambda^*$ are sorted in a non-increasing order by their lengths, in which a smaller index indicates a path with a higher length in general, \ie, $len(\lambda_l) \geq len(\lambda_h) \geq len(\lambda_s)$ with $1\leq l <h<s \leq |\Lambda^*|$.
\begin{equation}\label{eqs:target}
 P(\lambda_h) = 1 - \sum_{l=1}^{h-1}P(\lambda_l) - \sum_{s=h+1}^{|\Lambda^*|}P(\lambda_s) 
\end{equation}

Following Eq.~\ref{eqs:target}, Fig.~\ref{figs:eq} illustrates the computation process for $P(\lambda_h)$ of every $\lambda_h \in \Lambda^*$.
Starting from the first (\ie, longest) path in $\Lambda^*$, we obtain $P(\lambda_h)$ by determining the following two values.
\begin{enumerate}[label=(\roman*)]
\item $\sum_{l=1}^{h-1}P(\lambda_l)$: sum of $P(\lambda_l)$ with $1 \leq l < h$ (Fig.~\ref{figs:eq}\blackcircle{$a$}); 
\item $\sum_{s=h+1}^{|\Lambda^*|}P(\lambda_s)$: sum of $P(\lambda_s)$ with $h < s \leq |\Lambda^*|$ (Fig.~\ref{figs:eq}\blackcircle{$b$}).
\end{enumerate}


For $\sum_{l=1}^{h-1}P(\lambda_l)$, it can be obtained directly since the computation process always starts from $\lambda_1 \in \Lambda^*$ in order. Hence, when computing $P(\lambda_h)$, all the $P(\lambda_l)$ with $1 \leq l < h$ have already been calculated in the previous steps, as shown in Fig.~\ref{figs:eq}.
As for $\sum_{s=h+1}^{|\Lambda^*|}P(\lambda_s)$, it is not determined when $P(\lambda_h)$ is under computation, in which any $P(\lambda_s)$ with $h < s \leq |\Lambda^*|$ has not been examined yet. 

However, we note that $\sum_{s=h+1}^{|\Lambda^*|}P(\lambda_s)$ is also involved in the probability where $\lambda_h$ is not executed (denoted as $P(\overline{\lambda_h})$), as shown in Eq.~\ref{eqs:overline}. As illustrated in Fig.~\ref{figs:eq}, there are two cases in which $\lambda_h$ is not executed among all possible situations:
(i) a longer path $\lambda_l$ is executed as the longest path while $\lambda_h$ is not executed, denoted as $P(\lambda_l \oplus \overline{\lambda_h})$ in Fig.~\ref{figs:eq}\blackcircle{$a_2$} and
(ii) a shorter path $\lambda_s$ is executed as the longest path, where $\lambda_h$ must not be executed, \ie, $P(\lambda_s)$ in Fig.~\ref{figs:eq}\blackcircle{$b$}.
\begin{equation}\label{eqs:overline}
P(\overline{\lambda_h}) =  \sum_{l=1}^{h-1}P(\lambda_l \oplus \overline{\lambda_h}) + 
\sum_{s=h+1}^{|\Lambda^*|}P(\lambda_s) 
\end{equation}


\begin{figure}[tp!]
\centering
\includegraphics[width=\linewidth]{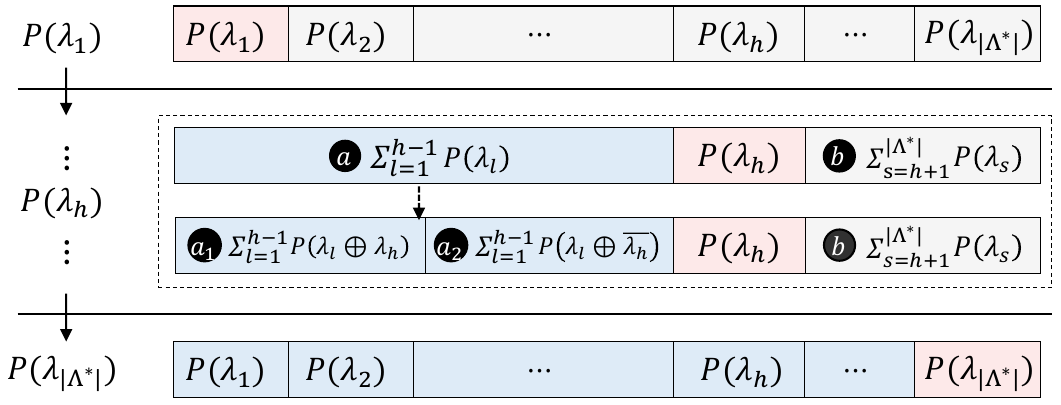} 
\caption{Computation of $P(\lambda_h)$ for every $\lambda_h \in \Lambda^*$ ordered in non-increasing path length \textit{(blue: probabilities that are examined; red: the probability under computation; grey: probabilities to be examined)}.}
\vspace{-10pt}
\label{figs:eq}
\end{figure}

In addition, given that the probability of $\lambda_h$ being executed is $\prod_{\theta_x^k \in H(\lambda_h)}F(\theta_x^k)$, the $P(\overline{\lambda_h})$ can also be computed by $P(\overline{\lambda_h}) = 1 - \prod_{\theta_x^k \in H(\lambda_h)}F(\theta_x^k)$. Accordingly, combining this with Eq.~\ref{eqs:overline}, we have $\sum_{l=1}^{h-1}P(\lambda_l \oplus \overline{\lambda_h}) + 
\sum_{s=h+1}^{|\Lambda^*|}P(\lambda_s) = 1 - \prod_{\theta_x^k \in H(\lambda_h)}F(\theta_x^k)$
Therefore, 
$\sum_{s=h+1}^{|\Lambda^*|}P(\lambda_s) $ can be computed as Eq.~\ref{eqs:shortp}.
\begin{equation}\label{eqs:shortp}
\sum_{s=h+1}^{|\Lambda^*|}P(\lambda_s)  = 1- \prod_{\theta_x^k \in H(\lambda_h)}F(\theta_x^k) - \sum_{l=1}^{h-1}P(\lambda_l \oplus \overline{\lambda_h})
\end{equation}

To this end, $\sum_{s=h+1}^{|\Lambda^*|}P(\lambda_s)$ can be obtained if $\sum_{l=1}^{h-1}P(\lambda_l \oplus \overline{\lambda_h})$ is computed.
However, it is difficult to compute 
$P(\lambda_l \oplus \overline{\lambda_h})$, which implies that all paths longer than $\lambda_l$ are not executed.
To bound $P(\lambda_l \oplus \overline{\lambda_h})$, we simplify the computation to only consider the case where $\lambda_l$ is executed while $\lambda_h$ is not, as shown in Eq.~\ref{eq:oplus}. As all branches in $H(\lambda_l)$ must be executed, $\theta_x^k \in H(\lambda_l) \cap  H(\lambda_h)$ are not included in the computation of $P(\overline{\lambda_h})$.
\begin{equation}\label{eq:oplus}
P(\lambda_l \oplus \overline{\lambda_h}) \leq
\prod_{\theta_x^k \in H(\lambda_l)} F(\theta_x^k)
\times 
\Big(1-\prod_{\theta_x^k \in H(\lambda_h) \setminus H(\lambda_l)} F(\theta_x^k)\Big) 
\end{equation}

To this end, $P(\lambda_l \oplus \overline{\lambda_h})$, $\sum_{s=h+1}^{|\Lambda^*|}P(\lambda_s)$, and eventually, the $P(\lambda_h)$ can be obtained by Eq.~\ref{eq:oplus},~\ref{eqs:shortp}, and~\ref{eqs:target}, respectively. The computation starts the longest path in $\Lambda^*$ and produces $P(\lambda_h)$ for every $\lambda_h \in \Lambda^*$ with $|\Lambda^*|$ iterations.
However, as fewer paths are considered when computing $P(\lambda_l \oplus \overline{\lambda_h})$, an upper bound is provided for $P(\lambda_l \oplus \overline{\lambda_h})$ according to the \textit{inclusion-exclusion principle}, instead of the exact value.
Hence, deviations can exist in $P(\lambda_h)$ as it depends on both $\sum_{l=l}^{h-1}P(\lambda_l)$ and $\sum_{s=h+1}^{|\Lambda^*|}P(\lambda_s)$. 
As a result, the probabilities of the long paths in $\Lambda^*$ could be overestimated, and subsequently, leading to
a lower probability for the shorter ones. 

Considering such deviations, the following constraints are applied for $P(\lambda_h)$ of every $\lambda_h \in \Lambda^*$: (i) $P(\lambda_h) \geq 0$ and (ii) $\sum_{l=1}^{h-1}  P(\lambda_l) + P(\lambda_h) \leq 1$.
First, we enforce that $P(\lambda_h) = 0$ if a negative value is obtained.
Second, if $\sum_{l=1}^{h-1}P(\lambda_l) + P(\lambda_h) > 1$ when computing $P(\lambda_h)$, the computation terminates directly with $P(\lambda_h) = 1 - \sum_{l=1}^{h-1}P(\lambda_l)$ and $P(\lambda_s) = 0,~h < s \leq |\Lambda^*|$.
Below we show that the deviations in $P(\lambda_h)$ would not significantly affect the following computations and prove the correctness of the analysis. 
\subsection{Discussions of Deviation and Correctness}\label{sec:4c}




We first demonstrate that the deviation of $P(\lambda_h)$ would not propagate along with the computation of every $\lambda_h \in \Lambda^*$.
As described in Sec.~\ref{sec:4a}, the deviations caused by Eq.~\ref{eq:oplus} can impact the value of $P(\lambda_h)$ from two aspects:
(i) the \textit{direct deviation} from $P(\lambda_l \oplus \overline{\lambda_h})$ in Eq.~\ref{eq:oplus} (denoted as $E_h^{dir}$); and 
(ii) the \textit{indirect deviation} from the deviations of $P(\lambda_l), 1 \leq l < h$ in Eq.~\ref{eqs:target} (denoted as $E_h^{ind}$). 
First, $E_h^{dir}$ is not caused by deviations of any $P(\lambda_l), 1 \leq l < h$ (see Eq.~\ref{eq:oplus}). 
Below we focus on $E_h^{ind}$ to illustrate the propagation of deviations.

Let $E_h$ denote the deviation of $\lambda_h$, $E_h$ and $E_h^{ind}$ is computed by Eq.~\ref{eq:e} and~\ref{eqs:error_pre}, respectively. For $E_h^{ind}$, it is caused by the deviations of $P(\lambda_l)$ with $1 \leq l < h$, as shown in Eq.~\ref{eqs:target}.
\begin{equation} \label{eq:e}
E_h = E_h^{{dir}} - E_h^{{ind}}
\end{equation}
\begin{equation}\label{eqs:error_pre}
E_h^{ind}  = \sum_{l=1}^{h-1} E_l
\end{equation}

Based on Eq~\ref{eq:e} and~\ref{eqs:error_pre}, $E_h^{ind}$ can be obtained using both equations recursively, as shown in Eq.~\ref{eqs:tmp}.
First, $E_h^{ind}$ is computed as $\sum_{l=1}^{h-1} E_l^{dir} - \sum_{l=1}^{h-1} E_l^{ind}$ based on Eq.~\ref{eq:e} and~\ref{eqs:error_pre}, in which $\sum_{l=1}^{h-1} E_l^{ind}$ is equivalent to $\sum_{l=1}^{h-2} E_l^{ind} + E_{h-1}^{ind}$. 
Then, $E_{h-1}^{ind}$ can be obtained using Eq.~\ref{eqs:error_pre}. 
Finally, $E_h^{ind}$ is computed as $E_{h-1}^{dir}$.
\begin{equation}\label{eqs:tmp}
\begin{split}
E_h^{ind} & = \sum_{l=1}^{h-1} E_l^{dir} - \sum_{l=1}^{h-1} E_l^{ind} \\
& =\sum_{l=1}^{h-1} E_l^{dir} -  (\sum_{l=1}^{h-2} E_l^{ind} + E_{h-1}^{ind}) \\
& =\sum_{l=1}^{h-1} E_l^{dir} -\Big(\sum_{l=1}^{h-2} E_l^{ind} + (\sum_{l=1}^{h-2} E_l^{dir} -\sum_{l=1}^{h-2} E_l^{ind})\Big)\\
& = \sum_{l=1}^{h-1} E_l^{dir} - \sum_{l=1}^{h-2} E_l^{dir}=E_{h-1}^{dir} 
\end{split}
\end{equation}






With $E_h^{dir}$ obtained, $E_h = E_h^{dir} - E_h^{ind} = E_h^{dir} -E_{h-1}^{dir}$ based on Eq.~\ref{eq:e}. 
From the computations, it can be observed that $P(\lambda_h)$ is only affected by the deviationtation of $P(\lambda_h)$ itself and $P(\lambda_{h-1})$. This dsemonstrate the analysis effectively manages the pessimventing the propagation of deviations. 
In addition, such deviations would not undermine the correctness of the analysis.
Let $\mathbb{P}(\lambda_h) = \sum_{1\leq l <h} P(\lambda_l) + P(\lambda_h)$ denote the probability where the length of the longest path being executed is equal to or higher than $len(\lambda_h)$. Theorem~\ref{the:correct} describes the correctness of the proposed analysis.

\begin{theorem}\label{the:correct}
Let $\mathbb{P}^\S(\lambda_h)$ denote the exact probability of $\mathbb{P}(\lambda_h)$, it follows that $\mathbb{P}( \lambda_h) \geq \mathbb{P}^\S( \lambda_h), \forall \lambda_h \in \Lambda^*$.
\end{theorem}
\begin{proof}
Based on Eq.~\ref{eq:e} and~\ref{eqs:tmp}, $P(\lambda_h) = P^\S(\lambda_h) + E_h^{dir} - E_{h-1}^{dir}$ with derivations considered.
In addition, given that $\mathbb{P}( \lambda_h) = \sum_{l=1}^{h} P(\lambda_l)$, in which $P(\lambda_l)$ is further computed as $P^\S(\lambda_l) + E_l^{dir} - E_{l-1}^{dir}$. Therefore, $\mathbb{P}( \lambda_h)$ and $\mathbb{P}^\S( \lambda_h)$ follows that $\mathbb{P}( \lambda_h) = \mathbb{P}^\S( \lambda_h) + E_{h}^{dir}$. As $E_h^{dir}$ is non-negative, $\mathbb{P}( \lambda_h) \geq \mathbb{P}^\S( \lambda_h)$ holds for all $\lambda_h \in \Lambda^*$.
\end{proof}





\begin{figure*}[t!]
\vspace{-10pt}
\centering
\subfigbottomskip=1pt 
\subfigure[$p=6$, $|\Theta|=3$ and varied $psr$.]{\label{figs:exp1a}
\includegraphics[width=.26\linewidth]{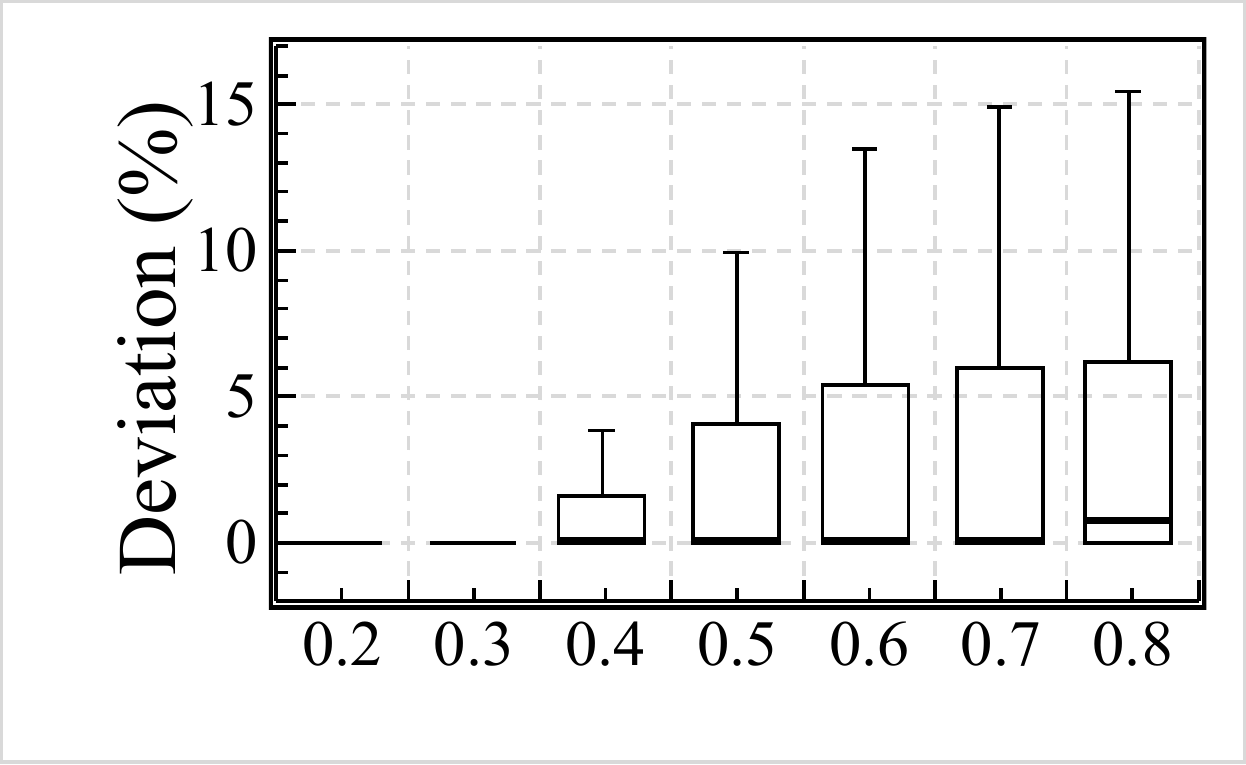}}\hspace{10pt}
\subfigure[$psr=0.4$, $|\Theta|=3$ and varied $p$.]{\label{fig:exp1b} 
\includegraphics[width=.26\linewidth]{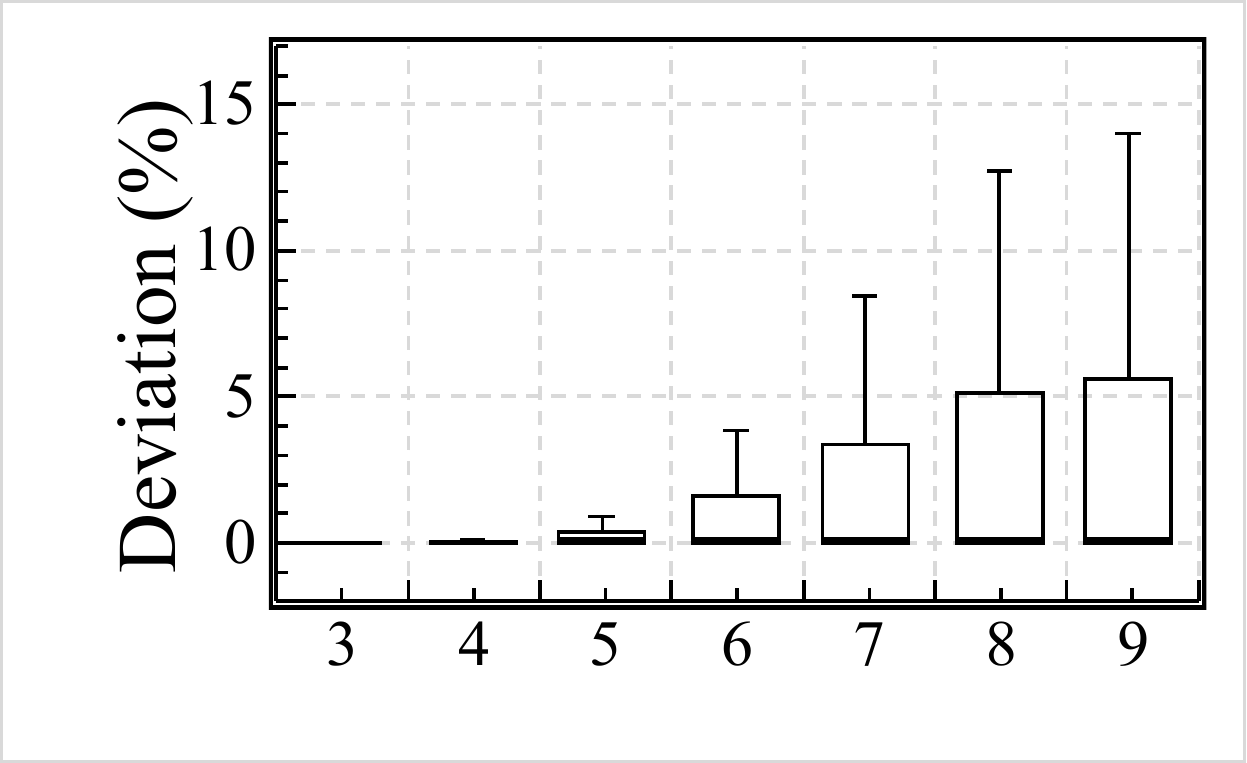}}\hspace{10pt}
\subfigure[$psr=0.4$, $p=6$ and varied $|\Theta|$.]{\label{fig:exp1c} 
\includegraphics[width=.26\linewidth]{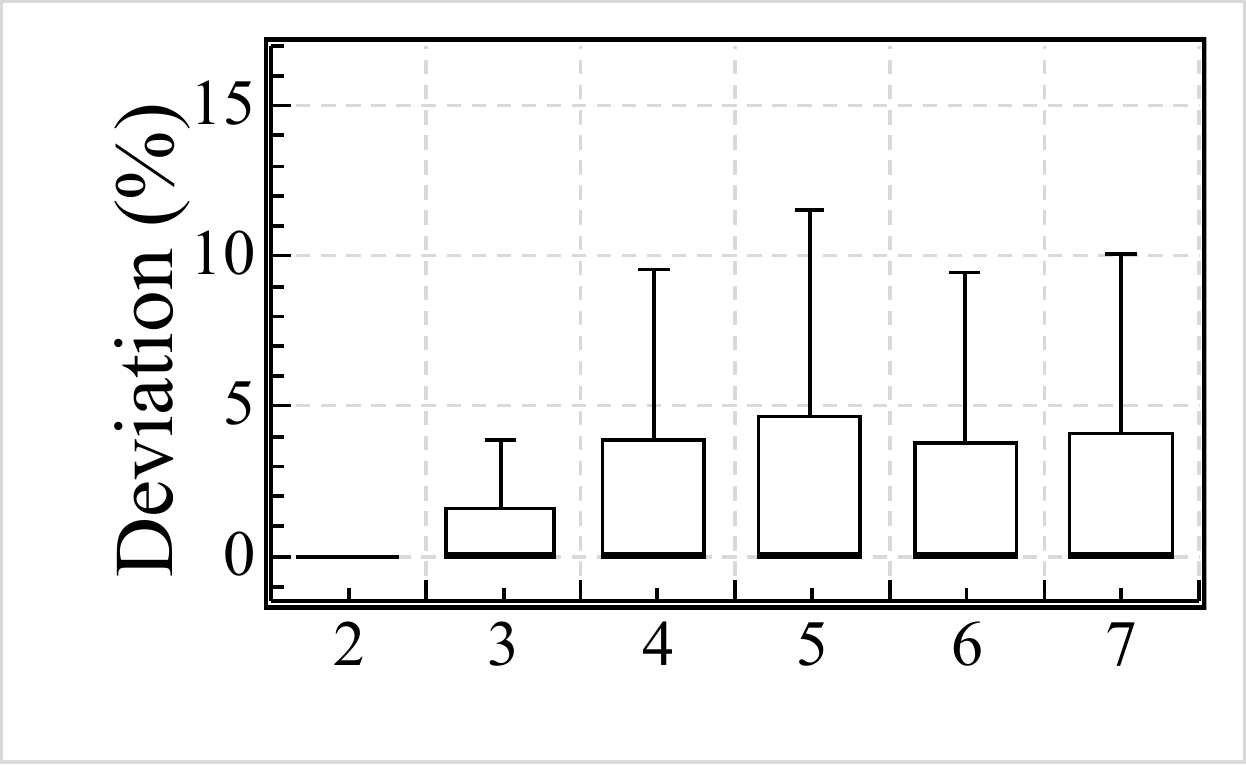}}
\caption{The deviation in percentage between the proposed analysis and Ueter2021 under varied $psr$, $p$ and $|\Theta|$. }
\label{figs:exp_one}
\vspace{-5pt}
\end{figure*}



\vspace{-4pt}
\section{Construction of Probabilistic Timing Distribution}\label{sec:interference}


With $P(\lambda_h)$ obtained for every $\lambda_h \in \Lambda^*$, this section constructs the complete probabilistic response time distribution of a \pdag{}.
When $\lambda_h$ is executed as the longest path, the worst-case interfering workload (\ie, denoted as $I_h$) is computed by Eq.~\ref{eqs:interference} in three folds: (i) the interference from the non-conditional nodes that are not in $\lambda_h$ (\ie, $v_i \notin \lambda_h  \wedge  v_i \notin \Theta$), (ii) the interference from nodes in $H(\lambda_h)$ that are not in $\lambda_h$ (\ie, $v_i \notin  \lambda_h \wedge v_i \in \bigcup_{\theta_x^k \in H(\lambda_h)}$), and (iii) the worst-case interference from nodes in probabilistic structures except $S(\lambda_h)$ (\ie, $\theta_x \in \Theta \setminus S(\lambda_h)$), where the branch with the maximum workload is taken into account.
\begin{equation}\label{eqs:interference}
\begin{split}
I_h  = & \sum_{v_i \notin \lambda_h  \wedge  v_i \notin \Theta} C_i + 
\sum_{ v_i \notin  \lambda_h \wedge v_i \in \bigcup_{\theta_x^k \in H(\lambda_h)}} C_i \\
+ & \sum_{\theta_x \in \Theta \setminus S(\lambda_h)} \mathop{max}\{ vol(\theta_x^k) ~|~ \theta_x^k \in \theta_x \} 
\end{split}
\end{equation}



With $I_h$, the worst-case response time $R_h$ when $\lambda_h$ is the longest path can be computed by $R_h \leq len(\lambda_h) + \frac{1}{m} \times I_h$, where $m$ denotes the number of cores. Hence, combing with $P(\lambda_h)$, the probabilistic response time distribution of $\tau$ can be obtained based on $\mathbb{P}(R_h) = \sum_{1\leq l <h} P(\lambda_l) + P(\lambda_h)$, in which $\mathbb{P}(R_h)$ indicates the probability where the response time of $\tau$ is equal to or higher than $R_h$.
Theorem~\ref{the:rta_correct} justifies the correctness of the constructed analysis for \pdag{}s. 

\begin{theorem}\label{the:rta_correct}
Let $R^\S_h$ and $\mathbb{P}^\S(R_h)$ denote the exact values of $R_h$ and $\mathbb{P}(R_h)$, it follows $R_h \geq  R^\S_h \wedge \mathbb{P}(R_h) \geq \mathbb{P}^\S(R_h), \forall \lambda_h \in \Lambda^*$.
\end{theorem}
\begin{proof}
We first prove that $R_h \geq  R^\S_h$ for a given $\lambda_h \in \Lambda^*$. When $\lambda_h$ is executed as the longest path, the only uncertainty is the interfering workload from $\theta_x \in \Theta \setminus S(\lambda_h)$ (\ie, the third part in Eq.~\ref{eqs:interference}). 
Based on Eq.~\ref{eqs:interference}, $I_h$ is computed based on the maximum workload of each $ \theta_x \in \Theta \setminus S(\lambda_h)$. As only one $\theta_x^k \in \theta_x$ is executed, this bounds the worst-case interfering workload of all possible scenarios; and hence, $R_h \geq  R^\S_h$ follows.
As for $\mathbb{P}(R_h) \geq \mathbb{P}^\S(R_h)$ for a given $R_h$, this is proved in Theorem~\ref{the:correct} where $\mathbb{P}(R_h) = \mathbb{P}(\lambda_h) \geq \mathbb{P}^\S(\lambda_h) = \mathbb{P}^\S(R_h)$.
Therefore, this theorem holds for all $\lambda_h \in \Lambda^*$.
\end{proof}

This concludes the constructed timing analysis for a \pdag{}. 
By exploiting the set of the longest paths in $\tau$ (\ie, $\Lambda^*$), the analysis produces the probabilistic response time distribution of $\tau$ (\eg, the one in Fig.~\ref{fig:1b}) without the need for enumerating through every $\mathcal{G}_u \in \mathcal{G}$. With a specified time limit (such as the failure rate defined by ISO-26262), the analysis can provide the probability where the system misses its deadline, offering valuable insights that effectively empower optimised system design solutions, \eg, the improved resource efficiency shown in Tab.~\ref{tab:require_cores} below.

\vspace{-1pt}
\section{Evaluation}\label{sec:exp}


This section evaluates the proposed analysis for \pdag{}s against the existing approaches~\cite{ueter2021response,graham1969bounds} in terms of deviations of analytical results (Sec.~\ref{sec:pessimism}), computation cost (Sec.~\ref{sec:cost}) and the resulting resource efficiency of systems with \pdag{}s (Sec.~\ref{sec:core}). 

\vspace{-1pt}
\subsection{Experimental Setup} \label{sec:results-setup}
The experiments are conducted on randomly generated \pdag{}s with $m=4$.
The generation of a \pdag{} starts by constructing the DAG structure. As with~\cite{zhao2020dag,zhao2022dag,he2019intra},
the number of layers is randomly chosen in $[ 5,8 ]$ and the number of nodes in each layer is decided in $[2,p]$ (with $p = 6$ by default). 
Each node has a 20\% likelihood of being connected to a node in the previous layer.
As with~\cite{jiang2024cache}, the period $T$ is randomly generated in $[1, 1400]$ units of time with $D=T$. The workload is calculated by $T \times 50\%$, given a total utilisation of $50\%$. The WCET of each node is uniformly generated based on the workload.
Then, a number of probabilistic structures (\ie, $|\Theta|$) are generated by replacing nodes in the generated DAG ($ |\Theta| = 3$ by default).
Each $\theta_x \in \Theta$ contains three probabilistic branches. A $\theta_x^k \in \theta_x$ is a non-conditional sub-graph generated using the same approach, with the number of layers and nodes in each layer determined in $[2,4]$.
A parameter \textit{probabilistic structure ratio} ($psr$) is used to control the volume of the probabilistic structures in $\tau$, \eg, $psr=0.4$ 
means the workload of probabilistic structures is 40\% of the total workload of $\tau$. 
The $F(\theta_x^k) \in [0,1]$ of each $\theta_x^k$ is assigned with a randomly probability, with $\sum_{\theta_x^k \in \theta_x} F(\theta_x^k)=1$ enforced for all $\theta_x^k$ in every $\theta_x$.
For each system setting, 500 \pdag{}s are generated to evaluate the competing methods.

The \textit{Non-Overlapping Area Ratio} (\metric)~\cite{ye2015cumulative} is applied to compare the probabilistic distributions produced by the proposed and enumeration-based~\cite{ueter2021response} (denoted as Ueter2021) analysis. It is computed as the non-overlapping area between the two distributions divided by the area of the distribution from Ueter2021. The area of a probabilistic distribution is quantified as the space enclosed by its distribution curve and the x-axis from the lowest to the highest response time. The non-overlapping area of two distributions is the space covered exclusively by either distribution. A lower value of \metric indicates a smaller deviation between the two distributions.

\begin{figure}[t]
\vspace{-12pt}
\centering
\includegraphics[width=.7\linewidth]{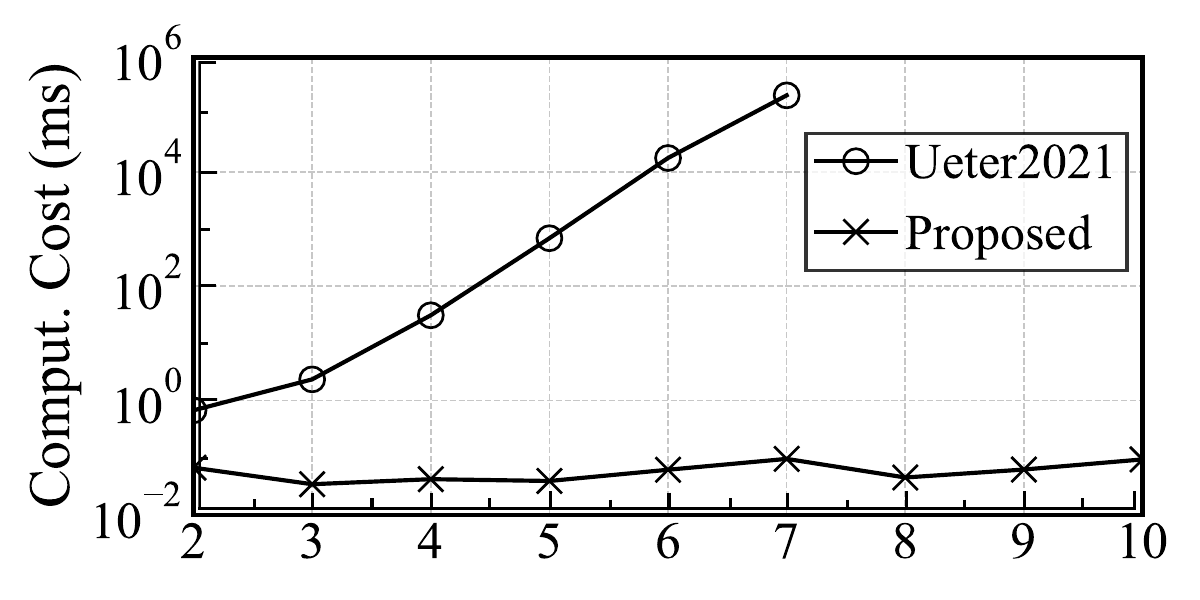}
\caption{Comparison of the computation cost under varied  $|\Theta|$.}
\label{figs:exp2}
\vspace{-10pt}
\end{figure}

\subsection{Deviations between the Analysis for \pdag{}s} \label{sec:pessimism}

This section compares the deviations between the proposed analysis and Ueter2021 in terms of \metric, as shown in
Fig.~\ref{figs:exp1a} to~\ref{fig:exp1c} with \pdag{}s generated under varied $psr$, $p$ and $|\Theta|$, respectively.

\textbf{Obs 1.}
The proposed method achieves an average deviation of only $1.04\%$ compared to Ueter2021, and remains below $5\%$ in most cases.

This observation is obtained from Fig.~\ref{figs:exp_one}, in which the deviation between our analysis and Ueter2021 is $1.45\%$, $0.73\%$ and $0.71\%$ on average with varied $psr$, $p$ and $|\Theta|$, respectively.
Notably, our method shows negligible deviations ($0.23\%$ on average) for \pdag{}s with a relatively simple structure, \eg, $psr \leq 0.4$, $p \leq 6$ or $|\Theta| \leq 3$.
As the structure of \pdag{}s becomes more complex, a slight increase in the deviation is observed between the two analysis, \eg, $3.01\%$ and $1.81\%$ on average when $psr=0.7$ in Fig.~\ref{figs:exp1a} and $p=8$ in Fig.~\ref{fig:exp1b}, respectively.
The reason is that for large and complex \pdag{}s, deviations can exist in multiple $P(\lambda_h)$ values (see Eq.~\ref{eq:oplus}), leading to an increased \metric between two analysis.
However, as observed, the deviations are below $5\%$ for most cases across all experiments, which justifies the effectiveness of the proposed analysis.

\subsection{Comparison of the Computation Costs}\label{sec:cost}
Fig.~\ref{figs:exp2} shows the computation cost (in milliseconds) of the proposed analysis and Ueter2021 under \pdag{}s generated under varied $|\Theta|$. The results are measured on a desktop with an Intel i5-13400 processor running at a frequency of 2.50GHz and a memory of 24GB.

\textbf{Obs 2.}
The computation cost of the proposed analysis is reduced by six orders of magnitude on average compared to Ueter2021.

As shown in the figure, the computation cost of Ueter2021 grows exponentially as $|\Theta|$ increases, due to the recursive enumeration of every execution scenario of a \pdag{}~\cite{ueter2021response}. 
In particular, this analysis fails to provide any results when $|\Theta| > 7$, which encounters an out-of-memory error on the experimental machine. 
In contrast, by eliminating the need for enumeration, our analysis achieves a significantly lower computation cost (under 10 milliseconds in most cases) across all $|\Theta|$ values and scales effectively to \pdag{}s with $|\Theta| >7$. 
Combining Obs.1 (Sec.~\ref{sec:pessimism}), the proposed analysis maintains a low deviation for relatively small \pdag{}s while effectively scaling to large ones, providing an efficient solution for analysing \pdag{}s. 


\begin{table}[t]
\centering
\vspace{-3pt}
\caption{Comparison of the average number of cores required under different acceptance ratios and varied $|\Theta|$.}
\label{tab:require_cores}
\resizebox{\columnwidth}{!}{
\begin{tabular}{l | c c c c c c c}
\toprule
\multicolumn{1}{c|}{\textbf{$|\Theta|$}} & 3 & 4 & 5 & 6 & 7  & 8  & 9\\ 
\midrule 
Ueter2021-70\%  & 8.66 & 8.18 & 8.07 & 7.88 & 7.76 & - & - \\
Proposed-70\%  & 8.66 & 8.19 & 8.09 & 7.89 & 7.78 & 7.20 & 7.07\\ 
\midrule

Ueter2021-80\%  & 12.78 & 11.31 & 10.67 & 9.93 & 9.94 & - & - \\
Proposed-80\%  & 12.78 & 11.31 & 10.67 & 9.93 & 9.94 & 9.05 & 8.30\\ 
\midrule

Ueter2021-90\%  & 12.78 & 11.32 & 10.70 & 10.01 & 10.11 & - & - \\
Proposed-90\%  & 12.78 & 11.32 & 10.70 & 10.01 & 10.11 & 9.33 & 8.63\\ 
\midrule

Ueter2021-100\%  & 13.67 & 12.38 & 12.07 & 11.96 & 12.22 & - & - \\
Proposed-100\%  & 13.67 & 12.38 & 12.07 & 11.96 & 12.22 & 12.16 & 12.10\\ 
Graham-100\%  & 13.67 & 12.38 & 12.07 & 11.96 & 12.22 & 12.16 & 12.10\\
\bottomrule
\end{tabular}}
\vspace{-8pt}
\end{table}

\vspace{-1pt}
\subsection{Impact on System Design Solutions}\label{sec:core}
\vspace{-1pt}
This section compares the resource efficiency of design solutions produced by the proposed and existing analysis. 
Tab.~\ref{tab:require_cores} shows the average number of cores required to achieve a given acceptance ratio, decided by the competing methods, \eg, ``Proposed-80\%" means that the proposed analysis is applied with an acceptance ratio of 80\%.

\textbf{Obs 3.} The proposed analysis effectively enhances the resource efficiency of systems, especially for ones with large \pdag{}s.



As shown in the table, both probabilistic analysis outperform Graham's bound in a general case by leveraging the probabilistic timing behaviours of \pdag{}s.
For instance, with the acceptance ratio of 70\%, our analysis reduces the number of cores by $36.33\%$ compared to Graham's bound when $|\Theta|=7$.
More importantly, negligible differences are observed between our analysis and Ueter2021 for $|\Theta|\leq7$, whereas our analysis remains effective as $|\Theta|$ continues to increase. 
This justifies the effectiveness of the constructed analysis and its benefits in improving design solutions by exploiting probabilistic timing behaviours of \pdag{}s. 
In addition, we observe that fewer cores are needed as $|\Theta|$ increases. This is expected as the construction of the probabilistic structures would increase task parallelism without changing the workload, leading to \pdag{}s that are more likely to be schedulable within a given limit (see Sec.~\ref{sec:results-setup}).






\vspace{-1pt}
\section{Conclusion}\label{sec:conclusion}

This paper presents a probabilistic response time analysis for a \pdag{} by exploiting its longest paths. We first identify the longest path for each execution scenario of the \pdag{} and calculate the probability of its occurrence.
Then, the worst-case interfering workload of each longest path is computed to produce the complete probabilistic response time distribution. 
Experiments show that compared to existing approaches, our analysis significantly enhances the scalability by reducing computation cost while maintaining low deviation, facilitating the scheduling of large \pdag{}s with improved resource efficiency. 
The constructed analysis provides an effective analytical solution for systems with \pdag{}s, empowering optimised system design that fully leverages the probabilistic timing behaviours.


\balance
\bibliographystyle{IEEEtran}
\bibliography{ref}

\begin{thebibliography}{10}
\providecommand{\url}[1]{#1}
\csname url@samestyle\endcsname
\providecommand{\newblock}{\relax}
\providecommand{\bibinfo}[2]{#2}
\providecommand{\BIBentrySTDinterwordspacing}{\spaceskip=0pt\relax}
\providecommand{\BIBentryALTinterwordstretchfactor}{4}
\providecommand{\BIBentryALTinterwordspacing}{\spaceskip=\fontdimen2\font plus
\BIBentryALTinterwordstretchfactor\fontdimen3\font minus \fontdimen4\font\relax}
\providecommand{\BIBforeignlanguage}[2]{{%
\expandafter\ifx\csname l@#1\endcsname\relax
\typeout{** WARNING: IEEEtran.bst: No hyphenation pattern has been}%
\typeout{** loaded for the language `#1'. Using the pattern for}%
\typeout{** the default language instead.}%
\else
\language=\csname l@#1\endcsname
\fi
#2}}
\providecommand{\BIBdecl}{\relax}
\BIBdecl

\bibitem{zhao2020dag}
S.~Zhao, X.~Dai, I.~Bate, A.~Burns, and W.~Chang, ``{DAG} scheduling and analysis on multiprocessor systems: Exploitation of parallelism and dependency,'' in \emph{2020 IEEE Real-Time Systems Symposium (RTSS)}.\hskip 1em plus 0.5em minus 0.4em\relax IEEE, 2020, pp. 128--140.

\bibitem{zhao2023cache}
S.~Zhao, X.~Dai, B.~Lesage, and I.~Bate, ``Cache-aware allocation of parallel jobs on multi-cores based on learned recency,'' in \emph{Proceedings of the 31st International Conference on Real-Time Networks and Systems}, 2023, pp. 177--187.

\bibitem{he2019intra}
Q.~He, N.~Guan, Z.~Guo \emph{et~al.}, ``Intra-task priority assignment in real-time scheduling of {DAG} tasks on multi-cores,'' \emph{IEEE Transactions on Parallel and Distributed Systems}, vol.~30, no.~10, pp. 2283--2295, 2019.

\bibitem{he2021response}
Q.~He, M.~Lv, and N.~Guan, ``Response time bounds for {DAG} tasks with arbitrary intra-task priority assignment,'' in \emph{33rd Euromicro Conference on Real-Time Systems (ECRTS 2021)}.\hskip 1em plus 0.5em minus 0.4em\relax Schloss Dagstuhl-Leibniz-Zentrum f{\"u}r Informatik, 2021.

\bibitem{baruah2021feasibility}
S.~Baruah, ``Feasibility analysis of conditional {DAG} tasks is co-npnp-hard,'' in \emph{Proceedings of the 29th International Conference on Real-Time Networks and Systems}, 2021, pp. 165--172.

\bibitem{ueter2021response}
N.~Ueter, M.~G{\"u}nzel, and J.-J. Chen, ``Response-time analysis and optimization for probabilistic conditional parallel {DAG} tasks,'' in \emph{2021 IEEE Real-Time Systems Symposium (RTSS)}.\hskip 1em plus 0.5em minus 0.4em\relax IEEE, 2021, pp. 380--392.

\bibitem{melani2015response}
A.~Melani, M.~Bertogna, V.~Bonifaci, A.~Marchetti-Spaccamela, and G.~C. Buttazzo, ``Response-time analysis of conditional {DAG} tasks in multiprocessor systems,'' in \emph{2015 27th Euromicro Conference on Real-Time Systems}.\hskip 1em plus 0.5em minus 0.4em\relax IEEE, 2015, pp. 211--221.

\bibitem{nelis2016variability}
V.~N{\'e}lis, P.~M. Yomsi, and L.~M. Pinho, ``The variability of application execution times on a multi-core platform,'' in \emph{16th International Workshop on Worst-Case Execution Time Analysis (WCET 2016)}.\hskip 1em plus 0.5em minus 0.4em\relax Schloss-Dagstuhl-Leibniz Zentrum f{\"u}r Informatik, 2016.

\bibitem{cazorla2013proartis}
F.~J. Cazorla, E.~Qui{\~n}ones, T.~Vardanega, L.~Cucu, B.~Triquet, G.~Bernat, E.~Berger, J.~Abella, F.~Wartel, M.~Houston \emph{et~al.}, ``Proartis: Probabilistically analyzable real-time systems,'' \emph{ACM Transactions on Embedded Computing Systems (TECS)}, vol.~12, no.~2s, pp. 1--26, 2013.

\bibitem{graham1969bounds}
R.~L. Graham, ``Bounds on multiprocessing timing anomalies,'' \emph{SIAM journal on Applied Mathematics}, vol.~17, no.~2, pp. 416--429, 1969.

\bibitem{zhao2022dag}
S.~Zhao, X.~Dai, and I.~Bate, ``{DAG} scheduling and analysis on multi-core systems by modelling parallelism and dependency,'' \emph{IEEE transactions on parallel and distributed systems}, vol.~33, no.~12, pp. 4019--4038, 2022.

\bibitem{yi2024can}
\BIBentryALTinterwordspacing
H.~Yi, J.~Liu, M.~Yang, Z.~Chen, and X.~Jiang, ``Improved convolution-based analysis for worst-case probability response time of {CAN},'' 2024. [Online]. Available: \url{https://arxiv.org/abs/2411.05835}
\BIBentrySTDinterwordspacing

\bibitem{iso201126262}
I.~ISO, ``26262: Road vehicles-functional safety,'' 2018.

\bibitem{birch2013safety}
J.~Birch, R.~Rivett, I.~Habli, B.~Bradshaw, J.~Botham, D.~Higham, P.~Jesty, H.~Monkhouse, and R.~Palin, ``Safety cases and their role in {ISO} 26262 functional safety assessment,'' in \emph{Computer Safety, Reliability, and Security: 32nd International Conference, SAFECOMP 2013, Toulouse, France, September 24-27, 2013. Proceedings 32}.\hskip 1em plus 0.5em minus 0.4em\relax Springer, 2013, pp. 154--165.

\bibitem{palin2011iso}
R.~Palin, D.~Ward, I.~Habli, and R.~Rivett, ``{ISO} 26262 safety cases: Compliance and assurance,'' 2011.

\bibitem{davis2019survey}
R.~I. Davis and L.~Cucu-Grosjean, ``A survey of probabilistic timing analysis techniques for real-time systems,'' \emph{LITES: Leibniz Transactions on Embedded Systems}, pp. 1--60, 2019.

\bibitem{abella2014comparison}
J.~Abella, D.~Hardy, I.~Puaut, E.~Quinones, and F.~J. Cazorla, ``On the comparison of deterministic and probabilistic {WCET} estimation techniques,'' in \emph{2014 26th Euromicro Conference on Real-Time Systems}.\hskip 1em plus 0.5em minus 0.4em\relax IEEE, 2014, pp. 266--275.

\bibitem{bernat2002wcet}
G.~Bernat, A.~Colin, and S.~M. Petters, ``{WCET} analysis of probabilistic hard real-time systems,'' in \emph{23rd IEEE Real-Time Systems Symposium, 2002. RTSS 2002.}\hskip 1em plus 0.5em minus 0.4em\relax IEEE, 2002, pp. 279--288.

\bibitem{bozhko2023really}
S.~Bozhko, F.~Markovi{\'c}, G.~von~der Br{\"u}ggen, and B.~B. Brandenburg, ``What really is p{WCET}? a rigorous axiomatic proposal,'' in \emph{2023 IEEE Real-Time Systems Symposium (RTSS)}.\hskip 1em plus 0.5em minus 0.4em\relax IEEE, 2023, pp. 13--26.

\bibitem{houssam2020hpc}
Z.~Houssam-Eddine, N.~Capodieci, R.~Cavicchioli, G.~Lipari, and M.~Bertogna, ``The {HPC-DAG} task model for heterogeneous real-time systems,'' \emph{IEEE Transactions on Computers}, vol.~70, no.~10, pp. 1747--1761, 2020.

\bibitem{zhao2024risk}
J.~Zhao, D.~Du, X.~Yu, and H.~Li, ``Risk scenario generation for autonomous driving systems based on causal bayesian networks,'' \emph{arXiv preprint arXiv:2405.16063}, 2024.

\bibitem{maier2023causal}
R.~Maier, L.~Grabinger, D.~Urlhart, and J.~Mottok, ``Causal models to support scenario-based testing of {ADAS},'' \emph{IEEE Transactions on Intelligent Transportation Systems}, 2023.

\bibitem{ziccardi2015epc}
M.~Ziccardi, E.~Mezzetti, T.~Vardanega, J.~Abella, and F.~J. Cazorla, ``{EPC}: extended path coverage for measurement-based probabilistic timing analysis,'' in \emph{2015 IEEE Real-Time Systems Symposium}.\hskip 1em plus 0.5em minus 0.4em\relax IEEE, 2015, pp. 338--349.

\bibitem{he2023real}
Q.~He, J.~Sun, N.~Guan, M.~Lv, and Z.~Sun, ``Real-time scheduling of conditional {DAG} tasks with intra-task priority assignment,'' \emph{IEEE Transactions on computer-aided design of integrated circuits and systems}, 2023.

\bibitem{melani2016schedulability}
A.~Melani, M.~Bertogna, V.~Bonifaci, A.~Marchetti-Spaccamela, and G.~Buttazzo, ``Schedulability analysis of conditional parallel task graphs in multicore systems,'' \emph{IEEE Transactions on Computers}, vol.~66, no.~2, pp. 339--353, 2016.

\bibitem{jiang2024cache}
Z.~Jiang, S.~Zhao, R.~Wei, Y.~Gao, and J.~Li, ``A cache/algorithm co-design for parallel real-time systems with data dependency on multi/many-core system-on-chips,'' in \emph{Proceedings of the 61st ACM/IEEE Design Automation Conference}, 2024, pp. 1--6.

\bibitem{ye2015cumulative}
N.~Ye, J.~P. Walker, and C.~R{\"u}diger, ``A cumulative distribution function method for normalizing variable-angle microwave observations,'' \emph{IEEE Transactions on Geoscience and Remote Sensing}, vol.~53, no.~7, pp. 3906--3916, 2015.

\end{thebibliography}

\end{document}